\documentclass[12pt]{article}

\usepackage{apacite,latexsym,graphics,epsfig,amsmath,amssymb,tabularx,rotating,booktabs, color,tikz}
\usepackage{natbib}
\usepackage{enumerate}
\allowdisplaybreaks
\usepackage{chngcntr}
\usepackage{amsthm}
\usepackage{url}
\usepackage[mathscr]{eucal}
\usetikzlibrary{arrows}
\makeatletter
\pgfdeclareshape{datastore}{
  \inheritsavedanchors[from=rectangle]
  \inheritanchorborder[from=rectangle]
  \inheritanchor[from=rectangle]{center}
  \inheritanchor[from=rectangle]{base}
  \inheritanchor[from=rectangle]{north}
  \inheritanchor[from=rectangle]{north east}
  \inheritanchor[from=rectangle]{east}
  \inheritanchor[from=rectangle]{south east}
  \inheritanchor[from=rectangle]{south}
  \inheritanchor[from=rectangle]{south west}
  \inheritanchor[from=rectangle]{west}
  \inheritanchor[from=rectangle]{north west}
  \backgroundpath{
    \southwest \pgf@xa=\pgf@x \pgf@ya=\pgf@y
    \northeast \pgf@xb=\pgf@x \pgf@yb=\pgf@y
    \pgfpathmoveto{\pgfpoint{\pgf@xa}{\pgf@ya}}
    \pgfpathlineto{\pgfpoint{\pgf@xb}{\pgf@ya}}
    \pgfpathmoveto{\pgfpoint{\pgf@xa}{\pgf@yb}}
    \pgfpathlineto{\pgfpoint{\pgf@xb}{\pgf@yb}}
 }
}
\makeatother

\newenvironment{remark}[1][Remark:]{\begin{trivlist}
\item[\hskip \labelsep {\bfseries #1}]}{\end{trivlist}}

\newcommand{\bfs}[1]{\mbox{\boldmath$#1$}}

\usepackage[T1]{fontenc}
\theoremstyle{definition}
\newtheorem{definition}{Definition}[section]
\theoremstyle{remark}
\newtheorem{example}{Example}[section]
\theoremstyle{plain}
\newtheorem{theorem}{Theorem}[section]

\newtheorem{proposition}[theorem]{Proposition}
\newtheorem{corollary}[theorem]{Corollary}

\def\keywords#1{{\vskip4pt
\noindent
\hbox to59.5pt{{\it{Key}}\enspace {\it{words}}:\quad\hss}\vtop{\advance \hsize by -59.5pt
\leftskip=28pt \rightskip=0pt
\noindent\ignorespaces#1\vskip8pt}}}

\newcommand{\tr}{^{\prime}}

\def\bd#1{\mbox{\boldmath $#1$}}
\newcommand*\xbar[1]{%
  \hbox{%
    \vbox{%
      \hrule height 0.5pt 
      \kern0.5ex
      \hbox{%
        \kern-0.25em
        \ensuremath{#1}%
        \kern-0.1em
      }%
    }%
  }%
}

\newcommand{\mysetminusD}{\hbox{\tikz{\draw[line width=0.6pt,line cap=round] (3pt,0) -- (0,6pt);}}}
\newcommand{\mysetminusA}{\,\, {\mysetminusD} \,\,}

\begin{document}

\title{Hierarchical Aitchison-Silvey models for incomplete binary sample spaces}

\author{Anna Klimova \\
{\small{National Center for Tumor Diseases (NCT), Partner Site Dresden, and}}\\
{\small{Institute for  Medical Informatics and Biometry,}}\\ 
{\small{Technical University, Dresden, Germany} }\\
{\small \texttt{anna.klimova@nct-dresden.de}}\\
{}\\
\and 
Tam\'{a}s Rudas \\
{\small{Department of Statistics, E\"{o}tv\"{o}s Lor\'{a}nd University, Budapest, Hungary}}\\
{\small \texttt{trudas@elte.hu}}\\
}

 \date{}
\maketitle

\vspace{20mm}

\begin{abstract}{
Multivariate sample spaces may be incomplete Cartesian products, when certain combinations of the categories of the variables are not possible. Traditional log-linear models, which generalize independence and conditional independence, do not apply in such cases, as they may associate positive probabilities with the non-existing cells.  To describe the association structure in incomplete sample spaces, this paper develops a class of hierarchical multiplicative models  which are defined by setting certain non-homogeneous generalized odds ratios equal to one and are named after Aitchison and Silvey who were among the first to consider such ratios. These models are curved exponential families that do not contain an overall effect and, from an algebraic perspective, are non-homogeneous toric ideals.  The relationship of this model class with log-linear models and quasi log-linear models is studied in detail in terms of both  statistics and algebraic geometry. The  existence of maximum likelihood estimates and their properties, as well as the relevant algorithms are also discussed.}
\end{abstract}

\section{Introduction}\label{SectionMotivation}

In multivariate statistical analysis the joint sample space of the variables is usually the Cartesian product of the ranges of the individual variables, which may be Euclidean spaces or discrete (finite or countable) sets. There are, however, situations where the joint sample space has a different structure -- most importantly in the case of the analysis of register data. Registers contain records which list the features that characterize particular events. For example, in a registry of congenital malformations, an event is  when a child under a specific age threshold is diagnosed with a congenital malformation and the features are the abnormalities observed. Or, in a register of road traffic violations, the records describe the rules violated by the driver. The registers contain no records where none of the features is present. If there are $k$ features, the sample space is not the Cartesian product CP=$\{present, not \,\,present\}^k$, rather the incomplete product IP=$\{present, not \,\,present\}^k \setminus \{not \,\,present\}^k$.

While the combination of no feature present is not part of a register, such cases may exist and even their number may be known. For example, there are newborns without congenital abnormalities, their number is known, and the register data may be completed so that the sample space will be the full Cartesian product CP. On the other hand, an instance of driving or parking without a traffic rule being violated hardly makes any sense, and even if some definition was agreed on, the number of those instances could not be determined. Thus, the sample space is inherently incomplete in the case of a traffic violation register.

When the sample space is inherently incomplete, statistical modeling has to reflect upon this fact, in particular that in so called big data analysis, register data are often used (see, e.g., \cite{Eeden2019}, \cite{SNDS}), and the specific properties implied by having observed data in an IP are usually overlooked.
In the discrete case, the often applied multiplicative (or log-linear) models would associate a positive probability with the non-existing category combination, and therefore are not appropriate. This paper introduces multiplicative models on the IP and discusses many  of their properties.

When there are two features, the sample space has the structure shown in Table \ref{Tab1}. Obviously, the model of independence defined by the assumption that the odds ratio (OR) is equal to 1 does not apply.

The sample space depicted in Table \ref{Tab1}  also occurs in experimental design, when the combination of no features (no treatments) is not observed. For example, \cite{Kawamura1995} observed traps filled with sugarcane of fish or both as bait and recorded the number of swimming crabs entering the traps. To test the hypothesis of the independence of the effects of the two types of baits, one would need the number of those swimming crabs entering a trap with no bait. Without having included this experimental condition, one is left wondering whether or not crabs could be caught in a trap with no bait. If the answer is yes, then the model of independence of the effects of the two bait types is relevant, but cannot be tested based on the available data. If the answer is no, the model of independence cannot be applied, as it would associate a positive probability with the experimental condition of no bait. On the other hand, the models developed in this paper do apply in such situations, as well.

Section \ref{SectionPreliminaries}  discusses whether any of the properties of independence which is the simplest multiplicative model on the CP may be retained for the case of the IP, and gives a negative answer. Further, the concept of the Aitchison-Silvey ratio
$$
ASR(A,B) =   \frac{p_{11}}{p_{10}p_{01}} 
$$
is introduced. The models to be developed later in the paper are defined by restrictions on  ASRs and their generalizations.

\begin{table}
\centering
\begin{tabular}[t]{ccc}
& $B=0$& $B=1$\\
\hline \\ [-6pt]
$A=0$	& - & $p_{01}$\\
$A=1$ 	& $p_{10}$  &  $p_{11}$  \\[6pt]
\end{tabular}
\vspace{3mm}
\caption{Joint distribution of two binary features on IP.} 
\label{Tab1}
\end{table} 

Section \ref{SectionASRs} defines conditional and higher order ASRs, parallel to the usual definitions for ORs and also shows that they are the canonical parameters in an exponential family representation of all distributions on the IP. 
Section \ref{SectionMixedPar} considers mixed parameterizations of this exponential family on the IP using marginal distributions and higher order conditional ASRs. Also in this setup, a  variation independence property, up to a constant multiplier, holds between the two components.

\begin{table}
\centering        
\begin{tabular}{ccccc}
   & \multicolumn{2}{c}{$C=0$}&\multicolumn{2}{c}{$C=1$}\\
\cmidrule(lr){2-3}  \cmidrule(lr){4-5} 
   & $B=0$ & $B=1$ & $B=0$ & $B=1$ \\ [3pt]
\hline \\ 
$A=0$ &- & $p_{010}$ & $p_{001}$ & $p_{011}$\\
$A=1$ & $p_{100}$  & $p_{110}$ & $p_{101}$ & $p_{111}$\\
\end{tabular}
\vspace{5mm}
\caption{Joint distribution of three binary features on IP.} 
\label{Tab2}
\end{table}

When there are three binary features, $A$, $B$, $C$, the IP contains 7 cells and its structure is shown in Table \ref{Tab2}. 
If conditioned on $C=1$, there is a conditional odds ratio $COR(A,B|C=1)$ measuring the strength of association between $A$ and $B$, when $C$ is present. But as there is no $(0,0,0)$ cell, when conditioned on $C=0$, no COR exists. The existing modeling approaches in such situations include using context specific models  \citep[cf.][]{HosgaardCSImodels, NymanCSI}. Such a model would assume the independence of $A$ and $B$ in the context when $C=1$, but would make no restriction when $C=0$.

Another approach to analyzing IPs is applying quasi models \citep[cf.][]{Goodman68, Agresti83}. Such a model is usually defined by allowing exclusive multiplicative parameters in the cells where the model of interest is not assumed to hold, i.e., when $C=0$. This implies for log-linear models that the MLE reproduces those cell frequencies, and model misfit, if any, concentrates on the cells which are being modeled, i.e., when $C=1$. However, allowing exclusive parameters in the cells not of interest, does not imply that the MLE reproduces their observed frequencies in the more general model class of relational models, if the exponential family does not contain the overall effect.

Therefore, Section \ref{SectionNoOE} extends the definition of quasi models, as the collection of conditional distributions on the IP, when the model holds on the CP.    
This, however, does not affect the fact that a quasi model leaves nearly half of the sample space unmodeled and nearly half of the data do not enter the analysis based on the quasi model. The hierarchical AS (HAS) models proposed restrict, in addition to the CORs, also the conditional ASRs. For example, in the case of the sample space in (\ref{Tab2}), a conditional independence type model would assume that $COR(A,B|C=1)=1$ and also that conditional Aitchinson-Silvey ratio, 
$$
CASR(A,B|C=0) =   \frac{p_{110}}{p_{100}p_{010}},
$$
is equal to $1$.

The development of HAS models in Section \ref{SectionNoOE} is parallel to the usual treatment of hierarchical log-linear (HLL) models, (see, e.g., \cite{RudasSpringer1}) and comparisons are made to quasi log-linear (QLL) models. It is shown, among others, that QLL models are obtained from HAS models by allowing an overall effect and vice versa. \cite{KRoverEff} discussed in detail the effect of adding or removing the overall effect to/from an exponential family and some of the results presented in that paper are generalized here. While the treatment is parallel, there are many noticeable differences, as HAS models are shown to be relational models without the overall effect, which has wide ranging consequences for the model properties and maximum likelihood estimation, discussed later in the paper. The HAS variants of the conditional independence and no-highest-order interaction models are studies in detail.

Section \ref{SectionGeometry} looks at the relationship between HLL, QLL and HAS models using the approach of algebraic geometry, to complete the result of Section 5. Statistical models are characterized not only as varieties but also as their vanishing ideals, as the collection of all properties which may be used to characterize the models. It turns out that the relationship between the model classes discussed in the paper may be described by well established transformations studied in algebraic geometry. The main results obtained here are that the QLL is the elimination ideal of the HLL on the IP, and the HLL on the CP is the extension ideal of the QLL, HAS on the IP is obtained from the HLL through dehomogenization, while the transformation in the opposite dection is the so called projective closure. This creates a direct link between HLL and HAS models, while in statistical terms, the link only existed though the QLL models.

The most important consequences of HAS models being relational models without the overall effect is that MLEs reproduce the marginals of the interactions allowed in the model only up to a constant multiplier. Algorithmic issues are discussed in Section \ref{SectionModelSelection}. An existing R function \citep{gIPFpackage} may be used to obtain MLEs under HAS models.

\section{Preliminaries}\label{SectionPreliminaries}

This section considers some of the  difficulties of modeling distributions on IPs, which serve as motivation for the later developments in the paper. In particular, standard definitions of independence, which are equivalent in the  CP case, will be investigated from the perspective of application of independence to IPs. 

Assume that the joint distribution of two binary random variables, \emph{features}, $A$ and $B$ can be described using the complete Table \ref{Tab1complete}.

\begin{table}
\centering
\begin{tabular}[t]{ccc}
& $B=0$& $B=1$\\
\hline \\ [-6pt]
$A=0$	& $p_{00}$ & $p_{01}$\\
$A=1$ 	& $p_{10}$  &  $p_{11}$  \\[6pt]
\end{tabular}
\vspace{5mm}
\caption{Joint distribution of two binary features on CP.} 
\label{Tab1complete}
\end{table}

One way to define that $A$ and $B$ are independent is 
\begin{equation}\label{ind}
p_{ij} = p_{i.}p_{.j}, \,\, i,j = 0,1,
\end{equation}
where the dot means marginalization. This definition may appear to apply to the incomplete table where only $3$ out of the $4$ cells are present, as seen in Table \ref{Tab1}. 
However, in this case, one would have $p_{01}=p_{0.}p_{.1}=p_{01}p_{.1}$,
which can only hold if either $p_{01}=0$ or $p_{.1}=1$, contradicting the assumed positivity.

A seemingly weaker definition of independence for Table \ref{Tab1complete} is the assumption, that there exist parameters $\alpha_i$ and $\beta_j$, such that
\begin{equation}\label{mult}
p_{ij}= \alpha_i \beta_j, \,\, i,j = 0,1.
\end{equation}
For a complete table, this definition is equivalent to (\ref{ind}) \citep{RudasSpringer1}, but for an incomplete table, (\ref{mult}) would not be a restriction. For example, by taking
$$
\alpha_1 = 1, \beta_1=p_{11}, \beta_0=p_{10}, \alpha_0=p_{01}/p_{11},
$$
one obtains infinitely many choices of parameters for (\ref{mult}) to hold hold. In the case of $k$ features, one would still have $2^k$ parameters and $2^k-1$ cell probabilities, and the generalization of (\ref{mult}) would  not be a restriction either. Note that upon taking logarithms of both sides, any generalization of (\ref{mult}) would become a system of linear equations. 

The definition of independence in (\ref{ind}) may be seen as the combination of two requirements. One is the multiplicative structure (\ref{mult}), and the other one is, that the joint distribution has pre-specified one-way marginal distributions $p_{i.}$ and $p_{.j}$.  Because of this, for a complete table, an independent distribution may be seen as an extension of marginal probability distributions to the product space. For a general result regarding extension of probability distributions defined on a system of marginals, to a joint product space, see \cite{Kellerer1964}. The cylinder sets associated with either one of the features generate an algebra in the set-theoretical sense and arbitrary (probability) measures given on these may be extended to the generated algebra if and only if, any pair of nonempty sets from the two algebras have a non-empty intersection  \citep{Marczewski1948}. This property is called qualitative independence and a generalization to conditionally independent extensions was given by \cite{BartfaiRudas1988}. Qualitative independence does not hold for  the incomplete Table \ref{Tab1}. For example, the marginals $p_{0.}= 0.6$ , \,\, $p_{.0}=0.7$    
would imply that $p_{01}=0.6$ and $p_{10}=0.7$, and thus, cannot be extended.  

Independence on Table \ref{Tab1complete} can also be expressed using the odds ratio (OR), by requiring that 
\begin{equation}\label{OR2}
\frac{p_{11}p_{00}}{p_{01}p_{10}} = 1.
\end{equation}
Generalizations of the OR and of the constraint in (\ref{OR2}) may be used to define log-linear models, see \cite{BFH}, \cite{Agresti2002}, \cite{Fien2007}, among others.

In the case of the incomplete Table \ref{Tab1}, \cite*{KRD11} called two features independent if 
\begin{equation}\label{ASind}
\frac{p_{11}}{p_{01}p_{10}} = 1.
\end{equation}
This definition was later referred to as the Aitchison-Silvey independence, to reflect upon the connection with a statistical model analogous to (\ref{ASind}) for the case of three variables \citep{AitchSilvey60}, which is described in the next example: 

\begin{example}\label{ExampleAS3}
Let $A$, $B$, and $C$ be binary random variables, representing certain features of subjects in a population, and assume that all subjects possess at least one feature (no unaffected cases).  The sample space can therefore be described using Table \ref{Tab2}. 
Here $p_{ijk} > 0$ for all $i,j,k \in \{0,1\}$, except for $i=j=k=0$, are the probabilities of possessing the corresponding combinations of features, and $\sum_{i,j,k} p_{ijk} = 1$.   According to the definition of \cite{AitchSilvey60}, $A, B, C$ are independent if, for some positive parameters $\alpha, \beta,\gamma$: 
\begin{eqnarray}\label{Bliss3param}
p_{100} &=& \alpha, \quad p_{010} = \beta,\quad p_{001} = \gamma, \nonumber \\
 p_{110} &=& \alpha \beta, \,\,  p_{101} = \alpha \gamma, \,\,  p_{011} = \beta\gamma, \,\,  p_{111} = \alpha \beta \gamma,
\end{eqnarray}
or, equivalently:
\begin{eqnarray}\label{Bliss3prob}
\frac{p_{110}}{p_{100} p_{010}} = 1, \,\, \frac{p_{101}}{p_{100}p_{001}} = 1, \,\, \frac{p_{011}}{p_{010}p_{001}} = 1, \,\,  \frac{p_{111}}{p_{100}p_{010}p_{001}} = 1.
\end{eqnarray}
\qed
\end{example}

In later parts of this paper, a generalization of the ratios appearing in the left hand side of (\ref{ASind}) and (\ref{Bliss3prob}) will be proposed and further used to build models for incomplete tables. The ratios will be called Aitchison-Silvey ratios (ASR), to emphasize the connection with the previous work of these authors.

Another approach to obtain a model on an incomplete table is to use a model on the corresponding complete one. Quasi-log-linear models, and, in particular, quasi-independence, are traditionally defined as models for the complete tables, with additional and unique multiplicative parameters included in the cells which are absent in the incomplete table \citep[cf.][]{BFH, FienbergIncomplete1972,CloggShihadeh}. For the setup in Example \ref{ExampleAS3}, the quasi-log-linear model of complete independence would be 
\begin{equation}\label{quasi}
p_{ijk}= \gamma \delta_i \epsilon_j \iota_k, \,\, i, j, k = 0, 1, \,\, (i,j,k) \neq (0,0,0), 
\end{equation}
$$
p_{000}= \gamma \delta_i \epsilon_j \iota_k \kappa.
$$
The MLE of $p_{000}$ under such a model is equal to the observed probability in the cell $(0,0,0)$ \citep{BFH}. Therefore, augmenting the data with a zero in the cell $(0,0,0)$ entails that the MLE will be concentrated on the incomplete table, and is considered to be the MLE under the quasi-independence model.

In this paper, a quasi-model will be derived from the model on the complete table by conditioning on the corresponding incomplete sub-table. The equivalence of the definition using exclusive parameters to obtain exact fit in certain cells and of the definition based on conditioning is mentioned, for example, in \cite{CloggShihadeh}. The relationship between  the traditional log-linear models, the corresponding quasi-models, and models restricting AS ratios will be investigated in more detail in the forthcoming sections.

\section{Conditional and higher order  ASRs}\label{SectionASRs}

In this section, a recursive definition of  conditional AS ratios is given and their properties, which are comparable to properties  of the conditional odds ratios (CORs),  are discussed.  

Let $\mathcal{F} = \{ F_1, \ldots, F_k \}$ be a set of binary random variables taking values in the set $\{0,1\}$, equal to the indicators of certain features of subjects in a given population. Denote by $\mathcal{I}_0 = \{0,1\} \times \cdots \times \{0,1\}$ the Cartesian product of the ranges written as a vector of 0--1 sequences of length $k$. Assume further that every subject possesses at least one feature of interest, and therefore,  the sample space of $F_1, \dots, F_k$ is the incomplete Cartesian product  $\mathcal{I} = \mathcal{I}_0 \mysetminusA \{(0,\dots,0)\}$. In the sequel,  $\mathcal{I}_0$ and $\mathcal{I}$ are referred to as CP and IP respectively, each sequence $\boldsymbol i$ in either CP or IP will be called a cell, with $\boldsymbol i_0 =(0,\dots,0) \in \mathcal{I}_0$ being called the zero cell. Finally, assume that the joint distribution of $F_1, \dots, F_k$ is parameterized by cell probabilities $\boldsymbol p =(p(\boldsymbol i))$. The cardinality of $\mathcal{I}$ will be denoted by $I$: $|\mathcal{I}| = 2^k -1 = I$.

The Aitchison-Silvey ratio for the two-way IP described in Table \ref{Tab1}  
is the following function of $\boldsymbol p$:
\begin{equation}\label{ASR}
{ASR} = \frac{p_{11}}{p_{01} p_{10}}.
\end{equation} 
The higher order ASRs calculated in conditional tables are called conditional ASRs (CASR) and illustrated next. Consider a probability distribution on the three-way IP in Table \ref{Tab2}. By conditioning on $A=1$, one obtains a two-way CP, while conditioning on $A=0$ leads to a two-way IP. The odds ratio in the former one is the COR:
$$
{COR}(B,C|A=1) = \frac{p_{111}p_{100}}{p_{110}p_{101}}.
$$
and the ASR in the latter one is the CASR:
$$
{CASR}(B,C|A=0) = \frac{p_{011}}{p_{010}p_{001}}.
$$

Just like for conditional odds ratios, the following result is obtained by straightforward calculations:
\begin{equation}\label{common}
\frac{{COR}(B,C|A=1)}{{CASR}(B,C|A=0)} = \frac{{COR}(A,C|B=1)}{{CASR}(A,C|B=0)} = \frac{{COR}(A,B|C=1)}{{CASR}(A,B|C=0)}.
\end{equation}
This means, that the ratio of the COR to the CASR in a three-way IP, does not depend on which variable is in the condition, thus it reflects a property of all the three variables. The common value in (\ref{common}), is equal to
\begin{equation}\label{2oASR}
{ASR}(A, B, C) = \frac{p_{111}p_{100}p_{010}p_{001}}{p_{110}p_{101}p_{011}},
\end{equation}
and will be called the second order AS ratio. The numerator is the product of the probabilities of those cells where the number of indices equal to $1$ has the same parity as the number of features, and the denominator is the product of the probabilities of those cells, where the parity is different. Conventionally, the second order OR would use the parity of the $0$ indices in the same way:
$$
{OR}(A, B, C) = \frac{p_{000}p_{011}p_{101}p_{110}}{p_{111}p_{001}p_{010}p_{100}}
$$
Here, the ratios in (\ref{common}) and in (\ref{2oASR}) show the effect of having a feature present as opposed to not having it present.  Using these rather than the reciprocals seems more meaningful in the current context. 

Higher order ASRs are defined similarly, in a recursive way. Once the ASR was defined for a $(k-1)$-dimensional IP, choose a feature, say $F_s$, and partition the IP into the $(k-1)$-dimensional CP, obtained by conditioning on $F_s = 1$, and the $(k-1)$-dimensional IP, obtained by conditioning on $F_s = 0$.  It turns out that the ratio of the COR in the former table and the CASR in the latter one does not depend on the feature upon which the conditioning occurred. This ratio, to be called the $k$-way or $(k-1)$st order AS ratio for $F_1, \ldots, F_k$ on the IP $\mathcal{I}$, is equal to:
\begin{equation}\label{ASRdef}
{ASR}(F_1, \ldots, F_k)=\frac{\prod_{\boldsymbol i \in \mathcal{I}_{pk}}p({\boldsymbol i})}{\prod_{\boldsymbol i \in \bar{\mathcal{I}}_{pk}}p({\boldsymbol i})},
\end{equation}
where $\mathcal{I}_{pk}$ is the subset of $\mathcal{I}$ that comprises the cells where the parity of the $1$ indices is the same as the parity of $k$, and $\bar{\mathcal{I}}_{pk} = \mathcal{I} \mysetminusA \mathcal{I}_{pk}$. Notice that the parity of $1$'s  in the cell $(1, \dots, 1)$ is equal to the parity of $k$, thus the probability of this cell always appears in the numerator.

In order to give a formal description of conditional AS ratios, denote by $\mathcal{V}= 2^{\mathcal{F}}$ the power set of $\mathcal{F}$, and let $\phi: \mathcal{I} \rightarrow (\mathcal{V} \mysetminusA \{\bfs{\varnothing}\} )$ be the function, such that, for every $\boldsymbol i \in \mathcal{I}$,
\begin{align*}
\phi(\boldsymbol{i}) = \{F_{{i}_1}, \ldots,  F_{{i}_l} \} & \mbox{ if and only if the indices equal to } 1 \makebox{ in } \boldsymbol{i} \\
&\makebox{ are exactly in the positions } {i}_1, \ldots, {i}_l.
\end{align*}

To condition, for example, on the complement of $\phi(\boldsymbol i) $ being zero, one chooses only cells $\boldsymbol j$ with $\phi(\boldsymbol{j}) \subseteq \phi(\boldsymbol{i}) $,  and computes the ASR for these.
Denote by $\mathcal{I}_{pi}$ the set of cells, where the parity of indices equal to $1$ is the same as in cell $\boldsymbol i$ and let $\bar{\mathcal{I}}_{pi}$ be the rest of the cells of the IP. Then,
\begin{equation}\label{CASRdef}
{CASR}(\phi(\boldsymbol{i})| \mathcal{F}\mysetminusA \phi(\boldsymbol{i}) =0)=\frac{\prod_{\boldsymbol j \in \mathcal{I}_{pi}, \phi(\boldsymbol j) \subseteq \phi(\boldsymbol i)}p({\boldsymbol j})}{\prod_{\boldsymbol j \in \bar{\mathcal{I}}_{pi}, \phi(\boldsymbol j) \subseteq \phi(\boldsymbol i)}p({\boldsymbol j})},
\end{equation}

The following proposition describes the recursive nature of conditional AS ratios:

\begin{proposition} 
Assume that the total number of features $k \geq 3$, $|\phi(\boldsymbol{i})| > 2$, and $\tilde{F}$ be a subset of features such that: $\tilde{F}\subsetneq \phi(\boldsymbol{i})$, with $|\tilde{F}| = 1$.  Then,                                             
$${CASR}(\phi(\boldsymbol{i}) \mid \mathcal{F} \mysetminusA \phi(\boldsymbol{i})  = \boldsymbol 0) =\frac{ {COR}(\phi(\boldsymbol{i}) \mysetminusA \tilde{F} \mid \tilde{F}= 1, \,\mathcal{F} \mysetminusA \phi(\boldsymbol{i}) = \boldsymbol 0)} {{CASR}(\phi(\boldsymbol{i}) \mysetminusA \tilde{F} \mid \tilde{F} = 0, \,\mathcal{F} \mysetminusA \phi(\boldsymbol{i}) = \boldsymbol 0)}. $$\qed
\end{proposition}
\noindent The proof goes by induction and is omitted for the sake of brevity.

To see the role of the AS ratios from a different perspective, consider the set of all (positive) probability distributions on the $k$-way IP.  This set forms an exponential family, which is apparent, in particular,  from the  parameterization to be developed now.
Fix an order of the cells $\boldsymbol{i}$ of the IP and consider the subsets of the features $\phi(\boldsymbol{i})$ in the same order. Here and in the sequel, the reverse lexicographic order is used. Let $\mathbf{S}$ be the $(2^k-1) \times (2^k-1) $ matrix with entries
\begin{equation}\label{MatrixCornerGeneralt}
s_{\boldsymbol j\boldsymbol i} = \left\{ \begin{array}{cl} 1, & \mbox{ if } \,\,\phi(\boldsymbol j) \subseteq \phi(\boldsymbol i)\\
0, & \mbox{  otherwise}
\end{array}\right. \quad \mbox{for } \,\, \boldsymbol i, \boldsymbol j \in \mathcal{I}.
\end{equation}
\noindent Here, the row index $\boldsymbol j$ refers to a (non-empty) subset of $\mathcal{F}$, and the column index $\boldsymbol i$ refers to the cells in the IP $\mathcal{I}$.

\begin{example}\label{k=3t}
For the IP in (\ref{Tab2}), the reverse lexicographic order entails the correspondence:
$$
\begin{array}{llll}
(1,0,0) \leftrightarrow \{A \}, & (0,1,0) \leftrightarrow \{B \}, & (0,0,1) \leftrightarrow \{C \},& 
(1,1,0) \leftrightarrow \{A,B \}, \\
(1,0,1) \leftrightarrow \{A,C \},& (0,1,1) \leftrightarrow \{B,C \}, & (1,1,1) \leftrightarrow \{A,B,C \},  \\
\end{array}
$$
which implies that
\begin{equation}\label{cornerPar111}
\mathbf{S}=\left(\begin{array}{ccccccc}
1 & 0 & 0 & 1 & 1 & 0 & 1 \\ 
0 & 1 & 0 & 1 & 0 & 1 & 1 \\ 
0 & 0 & 1 & 0 & 1 & 1 & 1 \\ 
0 & 0 & 0 & 1 & 0 & 0 & 1 \\ 
0 & 0 & 0 & 0 & 1 & 0 & 1 \\ 
0 & 0 & 0 & 0 & 0 & 1 & 1 \\
0 & 0 & 0 & 0 & 0 & 0 & 1\\
\end{array}\right).
\end{equation}
\hfill{ } \qed
\end{example}

As the following result shows, the matrix $\mathbf{S}$ is invertible and thus may be used as a design matrix of a parameterization
\begin{equation}\label{paramt}
\log \boldsymbol p = \mathbf S \tr \boldsymbol{\beta},
\end{equation}
for some parameters $\boldsymbol{\beta}$. Such a parameterization is called a corner or effect parameterization \citep{RudasSpringer1}. 

If one uses the subsets of $\mathcal{F}$  as indices of the components of $\boldsymbol{\beta}$,  (\ref{paramt}) may be rewritten as 
\begin{equation}\label{param1t}
\log p(\boldsymbol i ) = \sum_{\boldsymbol j: \phi(\boldsymbol j)\subseteq \phi(\boldsymbol i)}  \beta_{\phi(\boldsymbol j)}.
\end{equation}

\begin{theorem} \label{TheoremParamt1} The following holds:
\begin{enumerate}[(i)]
\item The matrix $\mathbf{S}' = (s_{\boldsymbol i\boldsymbol j})$ is invertible;
\item The entries of $(\mathbf{S}\tr)^{-1}$ are equal to
\begin{equation}\label{invent}
\tilde{s}_{\boldsymbol i\boldsymbol j} = \left \{\begin{array}{cl}
(-1)^{|\phi({\boldsymbol i})\mysetminusA \phi({\boldsymbol j})|}, \,\, & \mbox{if } \,\, \phi({\boldsymbol j}) \subseteq \phi({\boldsymbol i})\\
0, & \mbox{  otherwise}.
\end{array}\right. \quad \mbox{for } \,\, \boldsymbol i, \boldsymbol j \in \mathcal{I}.
\end{equation}
\end{enumerate}
\end{theorem}
\begin{proof}

(i) From the definition of $s_{\boldsymbol j\boldsymbol i}$, $\mathbf{S}$ is an upper-triangular matrix, with the main diagonal consisting of $1$'s, so $det(\mathbf S) = 1$ and thus $\mathbf{S}$ and $\mathbf{S}'$ are invertible. 
 
(ii)  Part (i) implies that the parameter vector $\boldsymbol{\beta}$ is determined uniquely, once $\boldsymbol{p}$ is given.  Define
\begin{equation}\label{param3t}
\boldsymbol \gamma_{\phi({\boldsymbol i})} = \sum_{\boldsymbol j: \,\phi({\boldsymbol j}) \subseteq \phi({\boldsymbol i})} (-1)^{|\phi({\boldsymbol i}) \mysetminusA \phi({\boldsymbol j})|} \log \boldsymbol p(\boldsymbol j).
\end{equation}
Apply the Moebius inversion formula \citep[cf.][]{CameronCombi} to (\ref{param3t}), to obtain that 
\begin{equation}\label{param1t}
\log p(\boldsymbol i ) = \sum_{\boldsymbol j: \phi(\boldsymbol j)\subseteq \phi(\boldsymbol i)}  \gamma_{\phi(\boldsymbol j)}
\end{equation}
and then by (\ref{paramt}) and unicity, 
\begin{equation}\label{eqt}
\boldsymbol{\beta}=\boldsymbol{\gamma}.
\end{equation}

Let $\tilde{\mathbf{S}}$ be the $I \times I$ matrix with entries given in (\ref{invent}).
Then, equations (\ref{param3t}) and (\ref{eqt}) imply that  $\boldsymbol \beta = \tilde{\mathbf S} \log \boldsymbol p$, and thus $\tilde{\mathbf{S}} =(\mathbf{S}\tr)^{-1}$.

\end{proof}

Next,  the structure of the parameters defined in  (\ref{paramt}) is investigated. 

\begin{example}\label{k=2t}
For the IP in (\ref{Tab1}) one has: $\phi(1,0) = \{A\}$, \,\, $\phi(0,1) = \{B\}$, \,\, $\phi(1,1) = \{A,B\}$, 
\begin{equation}\label{AS2matrices}
\mathbf{S} = \left(\begin{array}{ccc} 1 & 0& 1\\
                                                     0 & 1 & 1\\
                                                     0 & 0 & 1\\ 
                            \end{array}\right), \qquad (\mathbf{S}\tr)^{-1}= \left(\begin{array}{rrr} 1 & 0& 0\\
                                                     0 & 1 & 0\\
                                                     -1 & -1 & 1\\ 
                            \end{array}\right).
                            \end{equation}
                            and therefore, $\boldsymbol \beta' = \left(\log p_{10}, \log p_{01}, \log({p_{11}}/({p_{10}p_{01}}))\right)$. 

\hfill{ } \qed
\end{example}

As the example shows, the parameters $\boldsymbol{\beta}$ can be the logarithms of either the cell probabilities or the ASRs. The result holds in general:

\begin{theorem} \label{TheoremParamt2} 
The parameters $\boldsymbol \beta =(\mathbf{S}\tr)^{-1} \log \boldsymbol p$ are the natural logarithms of  cell probabilities or of CASRs of the features in $\mathcal F$. 
\end{theorem}

\begin{proof}

For $k=2$, the previous example shows the claim is true.

If $k > 2$, Theorem \ref{TheoremParamt1} implies that 
\begin{equation}\label{betas}
\boldsymbol \beta_{\phi({\boldsymbol i})} = \sum_{\boldsymbol j: \,\phi({\boldsymbol j}) \subseteq \phi({\boldsymbol i})} (-1)^{|\phi({\boldsymbol i}) \mysetminusA \phi({\boldsymbol j})|} \log \boldsymbol p(\boldsymbol j).
\end{equation}
If $|\phi(\boldsymbol i)|=1$, then
$$
\boldsymbol \beta_{\phi({\boldsymbol i})} =  \log p({\boldsymbol i}).
$$
If $|\phi(\boldsymbol i)|=2$, then $\phi(\boldsymbol i) = \{F_{i_1}, F_{i_2}\}$, for some $i_1$ and $i_2$. Therefore, from (\ref{betas}), 
\begin{align*}
\boldsymbol \beta_{\phi({\boldsymbol i})} &=  \log p(\phi^{-1} (\{F_{i_1},F_{i_2}\}) - \log p(\phi^{-1} (F_{i_1})) -  \log p(\phi^{-1} (F_{i_2})) \\
&{}\\
& = \log \frac{p(\phi^{-1} (\{F_{i_1},F_{i_2}\}))}{ p(\phi^{-1} (F_{i_1}))p(\phi^{-1} (F_{i_2})) } 
=\log  {CASR}(F_{i_1},F_{i_2}\mid \mathcal{F} \mysetminusA\{F_{i_1},F_{i_2}\} = 0).
\end{align*}
Finally, let $|\phi(\boldsymbol i)| = m >2$. Then,
for all $\boldsymbol j$, such that $\phi({\boldsymbol j}) \subseteq \phi({\boldsymbol i})$, one has:
\begin{align*}
& |\phi({\boldsymbol i}) \mysetminusA \phi({\boldsymbol j})|  \mbox{ is even if and only if } \boldsymbol j \in \mathcal{I}_{pi}, \\
& |\phi({\boldsymbol i}) \mysetminusA \phi({\boldsymbol j})|  \mbox{ is odd if and only if } \boldsymbol j \in \bar{\mathcal{I}}_{pi},
\end{align*}
(the sets $\mathcal{I}_{pi}$ and $\bar{\mathcal{I}}_{pi}$ were defined above (\ref{CASRdef})).

Therefore, equation (\ref{betas}) implies that
\begin{align*}
\boldsymbol \beta_{\phi({\boldsymbol i})} &=   \log \frac{\prod_{\boldsymbol j: \,\phi({\boldsymbol j}) \subseteq \phi({\boldsymbol i}), \boldsymbol j \in \mathcal{J}_{\boldsymbol i, e}}p(\boldsymbol j)}{\prod_{\boldsymbol j: \,\phi({\boldsymbol j}) \subseteq \phi({\boldsymbol i}), \boldsymbol j \in \mathcal{J}_{\boldsymbol i, o}}p(\boldsymbol j)}.
\end{align*}
Hence, with respect to (\ref{ASRdef}) and (\ref{CASRdef}),
\begin{align*}
\boldsymbol \beta_{\phi({\boldsymbol i})} &=   \left\{\begin{array}{ll} \log {CASR} (\phi(\boldsymbol i) \mid \mathcal{F} \mysetminusA\phi(\boldsymbol i) = 0),  &\mbox{ if } \phi(\boldsymbol i) \subsetneq \mathcal{F},\\
& \\
\log {ASR} (\mathcal{F}), & \mbox{ if } \phi(\boldsymbol i) =\mathcal{F}.\\
\end{array}
\right.
\end{align*}

\end{proof}

\begin{remark}
The foregoing developments  are modeled on standard results for CPs and ORs. In that case, $\mathcal{I}$ is defined without the exclusion of the cell $(0, \ldots, 0)$ but definitions (\ref{MatrixCornerGeneralt}) and (\ref{paramt}) remain the same. Theorems \ref{TheoremParamt1} and \ref{TheoremParamt2} also hold, with the exception that in the latter one, instead of $CASR$s, one has $COR$s. For details, see \cite{RudasSpringer1}. In Example \ref{k=3t}, one would have an additional cell: $(0,0,0) \leftrightarrow \{\bfs{\varnothing}\}$, as the first one in the order, and the design matrix, say $\mathbf{T}$,  would have an additional row and an additional column, as compared to (\ref{cornerPar111}):
\begin{equation}\label{cornerPar11}
\mathbf{T} = \left(\begin{array}{cccccccc}
1 & 1 & 1 & 1 & 1 & 1 & 1 & 1\\ 
0 & 1 & 0 & 0 & 1 & 1 & 0 & 1 \\ 
0 & 0 & 1 & 0 & 1 & 0 & 1 & 1 \\ 
0 & 0 & 0 & 1 & 0 & 1 & 1 & 1 \\ 
0 & 0 & 0 & 0 & 1 & 0 & 0 & 1 \\ 
0 & 0 & 0 & 0 & 0 & 1 & 0 & 1 \\ 
0 & 0 & 0 & 0 & 0 & 0 & 1 & 1 \\
0 & 0 & 0 & 0 & 0 & 0 & 0 & 1\\
\end{array}\right)
\end{equation}
The additional row contains  $1$ in every position and corresponds to an additional parameter present in every cell of $\mathcal{I}$, called the overall effect. The additional column is for $p(0,0,0)$ and it only contains the overall effect. 
\hfill{ } \qed
\end{remark}

Exponential families of probability distributions may be parameterized not only with the canonical parameters $\mathbf{\beta}$ in (\ref{MatrixCornerGeneralt}), but also by mean value parameters obtained as 
\begin{equation}\label{mvpt}
\boldsymbol{\mu} = \mathbf{S}\boldsymbol{p}.
\end{equation}

The components of $\mathbf{S}\boldsymbol{p}$, called subset sums, are the sums of probabilities  of the cells characterized by a common subset of features.  In the standard CP case, the subset sums are marginal probabilities. The next example illustrates the structure of the mean value parameters entailed by the design matrix (\ref{cornerPar111}) .

\vspace{5mm}

\noindent\textbf{Example \ref{k=3t}} (revisited):
The components of $\boldsymbol{\mu} = (\mu_1, \mu_2, \dots, \mu_7)\tr$ are

$$\begin{array}{c}
p(1,0,0)+p(1,1,0)+p(1,0,1)+p(1,1,1) \\
p(0,1,0)+p(1,1,0)+p(0,1,1)+p(1,1,1)\\
p(0,0,1)+p(1,0,1)+p(0,1,1)+p(1,1,1)\\
p(1,1,0)+p(1,1,1)\\
p(1,0,1)+p(1,1,1)\\
p(0,1,1)+p(1,1,1)\\
p(1,1,1).\\
\end{array}
$$
Here, for example, $\mu_1$ is the sum of probabilities of the cells in which the feature $A$ is present, $\mu_6$ is the sum of probabilities of the cells in which the features $B$ and $C$ are both present.
\qed

\vspace{5mm}

As was seen in Theorem \ref{TheoremParamt1}, $\mathbf{S}$ is invertible, and thus, (\ref{mvpt}) is a parameterization, as
\begin{equation}\label{linequt}
\boldsymbol{p} = \mathbf{S}^{-1} \boldsymbol{\mu}.
\end{equation}
The range of feasible mean value parameters is obtained by imposing the normalization condition: $\boldsymbol 1\tr \mathbf{S}^{-1} \boldsymbol{\mu} = 1$ and the inequality constraints: $0 < \mathbf{S}^{-1} \boldsymbol{\mu} <1$.

As the next step, a mixed parameterization \citep[cf.][]{BrownBook, HJ2} of the exponential family of probability distributions  on the IP $\mathcal{I}$, defined in (\ref{paramt}), is constructed.

\section{Mixed  parameterizations generated by AS parameters}\label{SectionMixedPar}

Let $\mathcal{V} = 2^{\mathcal{F}}$ be the power set of $\mathcal{F}$. A class of subsets $\mathcal{F}_{asc} \subseteq {\mathcal{V}}$ is called ascending if for any $S \in \mathcal{F}_{asc}$ and any $T \in {\mathcal{V}}$, if $S \subseteq T$ then $T \in \mathcal{F}_{asc}$. A class of subsets is called descending if its complement with respect to $\mathcal{V}$ is an ascending class. Notice that a descending class is uniquely specified by its maximal (with respect to inclusion) elements, while  an ascending class is uniquely determined by its minimal elements  \citep{RudasSpringer1}.

Let $\mathcal{F}_{asc}$ be an ascending class of $\mathcal{F}$, and let $\mathcal{F}_{des} = \mathcal{V} \mysetminusA \mathcal{F}_{asc} \mysetminusA \{\bfs{\varnothing}\}$ be the complement descending class, excluding the empty set. Although, formally, $\mathcal{F}_{des}$ is not a descending class, it will be called like this for the simplicity of presentation. Consider the matrix $\mathbf{S}$ with entries (\ref{MatrixCornerGeneralt}) and its inverse $\mathbf{S}^{-1}$.
Denote by $\mathbf{A}$ the submatrix of $\mathbf{S}$, consisting of the rows corresponding to the  subsets of $\mathcal{F}$ comprising $\mathcal{F}_{des}$, and by  $\mathbf{D}$ the matrix formed by the {\textit{columns} of $\mathbf{S}^{-1}$ corresponding to the elements $\mathcal{F}_{asc}$}\footnote{The apparent reversal in the usage of the letters of A and D is introduced here to be in line with previous literature which will be referred to later on.}. Because $\mathbf{S} \cdot \mathbf{S}^{-1}$ is an identity matrix, it holds that $\mathbf{A}\mathbf{D}=\mathbf 0$, and thus, the columns of  $\mathbf{D}$ form a basis of the  kernel of  $\mathbf{A}$.

As will be shown in Theorem \ref{TwoDistr1}\textit{b},  the mapping
\begin{equation}\label{expandMix}
\boldsymbol p \longmapsto \left(\begin{array}{c} \mathbf{A} \boldsymbol{p} \\ \mathbf{D}' \log \boldsymbol{p}\end{array}\right)
\end{equation}
is a bijection, and therefore, the mean value parameters $\mathbf{A} \boldsymbol{p}$ and the canonical parameters $\mathbf{D}' \log\boldsymbol{p}$ create a (mixed) parameterization of the set of all positive distributions on the IP $\mathcal{I}$.

\vspace{5mm}

\noindent \textbf{Example \ref{k=3t}} (revisited):
Let $\mathbf{S}$ be the matrix in (\ref{cornerPar111}) and consider the descending class $\mathcal{F}_{des} = \{\{A\}, \{B\}, \{C\}\}$. It can be easily shown, that
\begin{equation}\label{inv}
\mathbf{S}^{-1}=\left(\begin{array}{ccccccc}
1 & 0 & 0 & -1 & -1 & 0 & 1 \\ 
0 & 1 & 0 & -1 & 0 & -1 & 1 \\ 
0 & 0 & 1 & 0 & -1 & -1 & 1 \\ 
0 & 0 & 0 & 1 & 0 & 0 & -1 \\ 
0 & 0 & 0 & 0 & 1 & 0 & -1 \\ 
0 & 0 & 0 & 0 & 0 & 1 & -1 \\ 
0 & 0 & 0 & 0 & 0 & 0 & 1 \\
\end{array}\right).
\end{equation}
The corresponding mixed parameterization is specified by the mean value parameters $\mathbf{A}\boldsymbol p$, and the canonical parameters $\mathbf{D}\tr \log  \boldsymbol p$, where
$$\mathbf{A} = \left(\begin{array}{ccccccc} 
1 & 0 & 0 & 1 & 1 & 0 & 1 \\ 
0 & 1 & 0 & 1 & 0 & 1 & 1 \\ 
0 & 0 & 1 & 0 & 1 & 1 & 1 \\ 
\end{array}\right), \qquad \mathbf{D}\tr = \left(\begin{array}{rrrrrrr}-1 & -1 & 0 & 1 & 0 & 0 & 0 \\ 
-1 & 0 & -1 & 0 & 1 & 0 & 0 \\ 
 0 & -1 & -1 & 0 & 0 & 1 & 0 \\ 
 1 & 1 & 1 & -1 & -1 & -1 & 1 \\
 \end{array}\right).$$
Here $\mathbf{A}$ is formed by the first three rows of $\mathbf{S}$, and $\mathbf{D}$ is the last four columns of $\mathbf{S}^{-1}$.   \qed

\begin{remark}
In the standard development of exponential families for the CP, the mean value and the canonical parameters entailed by a design matrix and its kernel basis, respectively, are variation independent. This is ensured by the (conventionally assumed) presence of the overall effect in the design matrix, which implies that the sum of the cell probabilities is one from among the mean value parameters, up to a reparameterization.  Under this assumption, any vector with the same subset sums as the given probability distribution is  a probability distribution as well.  However, if the overall effect is not present, a vector, whose subset sums are equal to those of the given distribution, does not necessarily sum to $1$, and moreover, there may be probability distributions whose subset sums are proportional to each other, as illustrated next:
\hfill{} \qed
\end{remark}

\vspace{5mm}

\noindent\textbf{Example \ref{k=2t}} (revisited):
Consider the mixed parameterization entailed by the first two rows of the matrix $\mathbf{S}$ given in (\ref{AS2matrices}): 
$$\mathbf{A} = \left(\begin{array}{ccc} 1 & 0& 1\\
                                                     0 & 1 & 1\\
                                \end{array}\right), \qquad 
\mathbf{D} = \left(\begin{array}{ccc} -1 & -1& 1\\                                                                                                                        \end{array}\right).$$
The row $(1,1,1)$ is not in row space of $\mathbf{A}$, so the overall effect is not present. The subset sums, $\mathbf{A}\boldsymbol p$, of the two probability distributions, say $\boldsymbol p_1, \, \boldsymbol p_2$, in Table \ref{T1t} are proportional to each other, $3\mathbf{A}\boldsymbol p_1 = 4\mathbf{A}\boldsymbol p_2$.  \qed
\vspace{5mm}

Therefore, because the overall effect is not present, the mean value and canonical parameters in (\ref{expandMix}) are not variation independent on the IP  in the conventional sense. To formulate and prove a related result in the case of IPs, the mean value parameters defined in (\ref{mvpt}) may be considered in the following extended form:
\begin{equation}\label{ExtPar}
\boldsymbol{\nu}'=(\nu_1, \,\, \boldsymbol{\nu}_2'), \,\, \mbox{such that } \,\,  \boldsymbol{\mu} = \nu_1 \boldsymbol{\nu}_2.
\end{equation}

The extended form makes it possible to consider distributions with proportional (but not necessarily equal) subset sums, as this will be needed in the sequel. Obviously, the components of $\boldsymbol{\nu}$ are not identified and values for which $\nu_1 \boldsymbol{\nu}_2$ 
are the same, determine the same $\boldsymbol{p}$. The mixed parameterization uses only some of the components of $\boldsymbol{\mu}$, those obtained from the rows of $\mathbf{A}$ belonging to the descending class. For example, for the distributions in Table \ref{T1t},  the $\boldsymbol{\nu}_2$ component of the extended mean value parameter may be selected as $(3,4)'$ for both $\boldsymbol p_1$ and $\boldsymbol p_2$, so in this case, $\nu_1$ is $4/20$ for $\boldsymbol p_1$, and $3/20$ for $\boldsymbol p_2$.

\begin{table}
\centering
\begin{tabular}[t]{ccc}
& $F_2=0$& $F_2=1$\\
\hline \\ [-6pt]
$F_1=0$	& - & 1/5\\
$F_1=1$ 	& 2/5  &  2/5 \\[6pt]
\end{tabular}
\begin{tabular}[t]{ccc}
& $F_2=0$& $F_2=1$\\
\hline \\ [-6pt]
$F_1=0$	& - & 8/20\\
$F_1=1$ 	&11/20  &  1/20 \\[6pt]
\end{tabular}
\caption{Two probability distributions with proportional subset sums}
\label{T1t}
\end{table}

The next theorem shows that the two components of the mixed parameterization are variation independent up to a constant multiplier. In other words, the canonical parameters and the extended mean value parameters are variation independent. 


\begin{theorem}\label{TwoDistr1}
Let $\boldsymbol q, \boldsymbol r$ be two positive probability distributions on the IP $\mathcal{I}$. There exists  a unique probability distribution $\boldsymbol p$ on $\mathcal{I}$, such that
\begin{enumerate}[(a)]
\item $\mathbf{A} \boldsymbol p = \gamma \mathcal{A} \boldsymbol q$ for some $\gamma > 0$, and  $\mathbf{D} \log \bd p = \mathbf{D}\log  \bd r$;  
\item If $\boldsymbol q= \boldsymbol r$, then $\boldsymbol p = \boldsymbol q = \boldsymbol r$;
\item $\gamma = 1$ if and only if $\boldsymbol p \equiv \boldsymbol q$.
\end{enumerate}
\end{theorem}
As it will be shown in the proof, the probability distribution $\boldsymbol p$ described above can be obtained as the pointwise limit of a sequence of frequency distributions derived using an iterative scaling procedure. 

\begin{proof}
\textit{(a)} The argument proceeds by construction, in which the distribution $\boldsymbol p$ is obtained as the limiting point of an iterative process consisting of two steps that are repeated until convergence. To begin, set $\gamma = 1$.  Then, for a current value of $\gamma$ iterate two steps:
\vspace{0.8mm}

\noindent \textit{Step 1}: Find the relaxation sequence of vectors that initiates from $\boldsymbol r$ and converges to a solution, say $\boldsymbol p_{\gamma}$, of $\mathbf{A}\boldsymbol \pi = \gamma \mathbf{A} \boldsymbol q$,  see \cite{Bregman} and \cite{KRipf1}.\\
\noindent \textit{Step 2}: If $\boldsymbol 1\tr\boldsymbol p_{\gamma} \neq 1$, update $\gamma$ using one of the methods described in \cite{KRipf1} for the algorithm G-IPF,  and go to \textit{Step 1}. \\
\vspace{0.4mm}

 It can be proven by induction that, for every feasible $\gamma > 0$, the limit of a relaxation sequence satisfies $\mathbf{D} \log \boldsymbol{p}_{\gamma} = \mathbf{D} \log \boldsymbol r$. The two steps above implement the G-IPF algorithm in \cite{KRipf1}, except for the choice of the initial point for the relaxation sequences. Because the proof of convergence of the original G-IPF does not rely on this choice, it applies to the current implementation and can be used to show that the iterative procedure described above converges to the unique probability distribution $\boldsymbol p^*$ that satisfies: $\mathbf{A} \boldsymbol p^* = \gamma^* \mathbf{A} \boldsymbol q$,  $\mathbf{D} \log \boldsymbol{p}^* = \mathbf{D} \log \boldsymbol r$, and $\boldsymbol 1\tr\boldsymbol p^* = 1$, for a unique $\gamma^* > 0$.

To prove \textit{(b)}, suppose $\boldsymbol q= \boldsymbol r$ and apply G-IPF. The algorithm converges to a unique limit, say $\tilde{\boldsymbol p}$, that satisfies: $\mathbf{A} \tilde{\boldsymbol p} =\tilde{\gamma} \mathbf{A} \boldsymbol q$, $\mathbf{D} \log \tilde{\boldsymbol p} = \mathbf{D} \log \boldsymbol q$, $\boldsymbol 1\tr \tilde{\boldsymbol p} = 1$, for a unique $\tilde{\gamma}$.  Because $\boldsymbol q$ already satisfies these equations with $\tilde{\gamma} = 1$, the uniqueness implies that $\tilde{\boldsymbol p} \equiv \boldsymbol q$.

The statement in \textit{(c)} is now obvious. 
\end{proof}

In the following section, a new model class on the IP which is specified by vanishing AS interactions is introduced, and its relationship with hierarchical log-linear models on the CP and their quasi variants on the IP is established.

\section{Hierarchical  AS and quasi log-linear models for the IP}\label{SectionNoOE}

Hierarchical AS (HAS) models on the IP are defined in the same manner as hierarchical log-linear models on the CP, except that the definition relies on conditional AS ratios instead of conditional odds ratios. 

\begin{definition}\label{HASmodel}
Let $\mathcal{F}_{asc}$ be an ascending class of subsets of $\mathcal{F}$. The hierarchical AS model generated by $\mathcal{F}_{asc}$ is the set of positive probability  distributions  on the IP $\mathcal{I}$ for which
\begin{eqnarray}\label{HASdef}
{CASR}(V | \mathcal{F} \mysetminusA V = 0 ) &=& {COR}(V | \mathcal{F} \mysetminusA V = 1 ) = 1, \quad \mbox{for all } V \in \mathcal{F}_{asc} \mysetminusA \mathcal{F}, \\ 
{ASR}(\mathcal{F}) &=& 1. \nonumber
\end{eqnarray}
\end{definition}
\vspace{1mm}

By Theorem \ref{TheoremParamt2}, the constraint in (\ref{HASdef}) is equivalent to requiring that:
\begin{equation}\label{HASbeta}
\beta_{\phi(\boldsymbol{i})} = 0 \makebox{ for all } \phi(\boldsymbol{i}) \in \mathcal{F}_{asc},
\end{equation}
and thus for distributions $\boldsymbol p$ in the model,
\begin{equation}\label{HASLLform}
\log p(\boldsymbol i ) = \sum_{\boldsymbol j: \phi(\boldsymbol j)\subseteq \phi(\boldsymbol i), \phi(\boldsymbol{j}) \notin \mathcal{F}_{asc}}  \beta_{\phi(\boldsymbol j)}.
\end{equation}
Further, the model (\ref{HASLLform}) can also be written in the matrix form:
\begin{equation}\label{HASmatrixform}
\log \boldsymbol p = \mathbf{A}\tr \boldsymbol \beta,
\end{equation}
where $\mathbf{A}$ is the matrix, whose rows are indicators of the subsets of $\mathcal{F}$ which comprise the descending class $\mathcal{F}_{des} = \mathcal{V} \mysetminusA \mathcal{F}_{asc} \mysetminusA \{\varnothing\}$, and $\boldsymbol \beta$ is a vector of  log-linear parameters which may be associated with those subsets. 
Naturally, $\mathbf{A}$ is a submatrix of $\mathbf{S}$ with entries given in (\ref{MatrixCornerGeneralt}). 

As follows from (\ref{HASmatrixform}), $HAS(\mathcal{F}_{asc})$ is a relational model generated by the design matrix $\mathbf{A}$ \citep{KRD11}. In the mixed parameterization constructed in Section \ref{SectionMixedPar} , $HAS(\mathcal{F}_{asc})$ can be specified by setting the canonical parameters corresponding to the kernel basis matrix $\mathbf{D}$ equal to zero. The distributions in the models are then parameterized by extended mean value parameters (\ref{ExtPar}).
	
The next example describes some HAS models  for three features.   
	
\vspace{5mm}

\noindent \textbf{Example \ref{k=3t}} (revisited):
Consider the HAS model generated by the ascending class $\mathcal{F}_{asc}= \{\{AB\}, \{ABC\}\}$. The design matrix is obtained from (\ref{cornerPar111}) by omitting the 4th and the 7th rows, corresponding to the $\{AB\}$ and $\{ABC\}$ subsets respectively:
\begin{equation}\label{modMxCIAS}
\mathbf{S}_{[AC][BC]} = \left(\begin{array}{ccccccc}
1&0&0&1&1&0&1\\
0&1&0&1&0&1&1\\
0&0&1&0&1&1&1\\
0&0&0&0&1&0&1\\
0&0&0&0&0&1&1\\
\end{array} \right).
\end{equation}
Here the notation $\mathbf{S}_{[AC][BC]}$ follows the convention for  hierarchical log-linear models, which are denoted  using the maximal elements of the generating descending class. 
The kernel basis matrix of $\mathbf{S}_{[AC][BC]}$ is
\begin{equation*}
\mathbf{D}\tr_{[AC][BC]} = \left(\begin{array}{rrrrrrr}
-1& -1& 0& 1& 0& 0&0 \\
0&0&1&0&-1&-1&1\\
\end{array} \right).
\end{equation*}
Equivalently, the model can be specified in terms of one non-homogeneous and one homogeneous odds ratio:
\begin{equation}\label{AScondOR}
\frac{p_{110}}{p_{100}p_{010}} = 1, \qquad  \frac{p_{001}p_{111}}{p_{101}p_{011}} =1.
\end{equation}
The first generalized odds ratio is non-homogeneous, and expresses the AS-independence of $A$ and $B$, given $C = 0$.  The second odds ratio expresses the  independence of $A$ and $B$, given $C =1$. The model can, therefore, be interpreted as a model of conditional AS independence. The model of assuming only the second restriction in  (\ref{AScondOR}) is usually referred to as context specific independence,  see, e.g., \cite{HosgaardCSImodels} and \cite*{NymanCSI}. It is clear from the definition, that, because the vector of $1$'s is not in the row space of $\mathbf{A}$,  the HAS conditional independence is a relational model without the overall effect. 

As the second example, consider the HAS model generated by the ascending class $\mathcal{F}_{asc}= \{\{ABC\}\}$, which can be interpreted as  the model of ``no second order AS interaction''. The model matrix is equal to 
\begin{equation}\label{modMxHomAss}
\mathbf{S}_{[AB][AC][BC]}=\left(\begin{array}{ccccccc}
1 & 0 & 0 & 1 & 1 & 0 & 1 \\ 
0 & 1 & 0 & 1 & 0 & 1 & 1 \\ 
0 & 0 & 1 & 0 & 1 & 1 & 1 \\ 
0 & 0 & 0 & 1 & 0 & 0 & 1 \\ 
0 & 0 & 0 & 0 & 1 & 0 & 1 \\ 
0 & 0 & 0 & 0 & 0 & 1 & 1 \\
\end{array}\right),
\end{equation}
and a kernel basis matrix of $\mathbf{S}_{[AB][AC][BC]}$ is 
\begin{equation}\label{kernelMHAS}
\mathbf{D}\tr_{[AB][AC][BC]} = \left(\begin{array}{rrrrrrr}
-1 & -1& -1&  1& 1& 1 & -1\\
\end{array} \right).
\end{equation}
In terms of the generalized odds ratios, the model can be expressed as:
$$
\frac{p_{110}p_{010}p_{001}}{p_{100} p_{010} p_{001}p_{111} } =1.
$$
As the latter can be rewritten as any of the three equations:
$$\frac{p_{110}}{p_{100}p_{010}} = \frac{p_{111}p_{001}}{p_{101}p_{011}}, \quad  \frac{p_{101}}{p_{100}p_{001}} = \frac{p_{111}p_{010}}{p_{110}p_{011}}, \quad \mbox{ or } \quad \frac{p_{011}}{p_{010}p_{001}} = \frac{p_{111}p_{100}}{p_{101}p_{110}},$$
one sees that, under this HAS model:
\begin{eqnarray*}
{CASR}(A,B|C=0) & ={COR}(A,B|C=1),\\
{CASR}(A,C|B=0) &={COR}(A,C|B=1),\\
{CASR}(B,C|A=0) &={COR}(B, C|A=1).
\end{eqnarray*}
This model is a relational model which can be interpreted as a model of homogeneous AS association. Just as for the conditional independence model above, the overall effect is not present. 
\qed

\vspace{7mm}

\begin{theorem}\label{theoremNoOE}
Let $\mathcal{F}_{asc} \subseteq \mathcal{V}$ be an ascending class of subsets of $\mathcal{F}$. 
If  $\mathcal{F}_{asc}\neq \{\bfs{\varnothing}\}$, the $HAS(\mathcal{F}_{asc})$ model  is a relational model without the overall effect. 
\end{theorem}

\begin{proof}
Let $\mathbf{A}$ be the model matrix of $HAS(\mathcal{F}_{asc})$  in the corner parameterization. 
Equation (\ref{HASbeta}) implies that under the model, 
 \begin{align*}
\beta_{\mathcal{F}} = \log {ASR}(\mathcal{F}) = \sum_{\boldsymbol j \in \mathcal{J}_{pk}}  \log p(\boldsymbol j)  - \sum_{\boldsymbol j \in \bar{\mathcal{J}}_{pk}}{ \log p(\boldsymbol j)} =0.
\end{align*}
which entails that the vector $\boldsymbol d = (d(\boldsymbol j))_{\boldsymbol j \in \mathcal{I}}$ with components
\begin{align*}
d(\boldsymbol j) = \left\{ \begin{array}{rc} 1, & \mbox{if } \, \boldsymbol j \in \mathcal{J}_{pk}\\
                                                              -1,  & \mbox{if } \, \boldsymbol j \in \bar{\mathcal{J}}_{pk}\\
                                                              \end{array}\right.
\end{align*}
belongs to $Ker(\mathbf{A})$. 

As $|\mathcal{J}_{pk}| \neq |\bar{\mathcal{J}}_{pk}|$, $\boldsymbol 1\tr \boldsymbol d \neq 0$ and thus the row space of $\mathbf{A}$ does not contain $\boldsymbol 1$.
\end{proof}

Hierarchical quasi log-linear models, described next, will be derived from conventional log-linear models on the CP, via conditioning those on the IP. As it will be seen, the quasi models can also be expressed using conditional odds ratios, however some of the cells in the IP will remain unconstrained under a quasi model.

\begin{definition}\label{HASmodel}
Let $\mathcal{F}_{asc}$ be an ascending class of subsets of $\mathcal{F}$, and let $LL(\mathcal{F}_{asc})$ be the log-linear model on the CP $\mathcal{I}_0$ generated by this class. The hierarchical quasi model $QLL(\mathcal{F}_{asc})$  is the following set of probability  distributions on the IP $\mathcal{I}$:  
$$QLL(\mathcal{F}_{asc})=\{\boldsymbol \pi = (\pi(\boldsymbol i))_{\boldsymbol i \in \mathcal{I}} > \boldsymbol 0:  \exists \boldsymbol p \in LL(\mathcal{F}_{asc}), \mbox{ such that } \,\, \pi(\boldsymbol i) = \frac{p(\boldsymbol i)}{\sum_{\boldsymbol i \in \mathcal{I}} p(\boldsymbol i)} \,\forall \,\boldsymbol i \in \mathcal{I}\}.$$
\end{definition}*
As follows from the definition, a  quasi log-linear model always has an overall effect.

The construction of a model matrix for a quasi conditional independence model is shown next.
\vspace{5mm}

\noindent\textbf{Example \ref{k=3t}} (revisited):
Consider the CP obtained from the IP in (\ref{Tab2}) by completing it with the $(0,0,0,)$ cell. The design matrix for the corner parameterization is equal to $\mathbf{T}$ in (\ref{cornerPar11}), see the Remark in Section \ref{SectionASRs}. 
The model matrix for the log-linear model of conditional independence $[AC][BC]$, generated by the ascending class $\mathcal{F}_{asc}= \{\{AB\}, \{ABC\}\}$ is derived by removing from $\mathbf{T}$ rows 5 and 8, corresponding  to the elements of $\mathcal{F}_{asc}$: 
\begin{equation*}\label{modMxCIAS2}
\mathbf{T}_{[AC][BC]} = \left(\begin{array}{cccccccc}
1 & 1 & 1 & 1 & 1 & 1 & 1 & 1\\ 
0 & 1 & 0 & 0 & 1 & 1 & 0 & 1 \\ 
0 & 0 & 1 & 0 & 1 & 0 & 1 & 1 \\ 
0 & 0 & 0 & 1 & 0 & 1 & 1 & 1 \\ 
0 & 0 & 0 & 0 & 0 & 1 & 0 & 1 \\ 
0 & 0 & 0 & 0 & 0 & 0 & 1 & 1 \\
\end{array}\right).
\end{equation*} 
For each distribution in the model, there exists a vector of parameters $\tilde{\boldsymbol \beta} = \tilde{\boldsymbol \beta}(\boldsymbol p)$ = ($\tilde{\beta}_0$, $\tilde{\beta}_1$, \dots,$\tilde{\beta}_5$)$\tr$, such that:
$$\log \boldsymbol p = \mathbf{T}_{[AC][BC]}\tr \tilde{\boldsymbol \beta},$$
or, equivalently, $\boldsymbol p = \exp\{\mathbf{T}_{[AC][BC]}\tr \tilde{\boldsymbol \beta}\}.$ The conditional distribution $\boldsymbol{\pi} = (\pi(\boldsymbol i))_{i \in \mathcal{I}}$ on the IP is obtained as
\begin{align}\label{qCI}
\pi(\boldsymbol i) = \frac{p(\boldsymbol i)}{\sum_{i \in \mathcal{I}} p(\boldsymbol i)},
\end{align}
and therefore,
\begin{align*}
 \pi_1 &= exp\{ \tilde \beta_1\}/Z, \,\, \pi_2 = exp\{ \tilde\beta_2\}/Z, \,\, \pi_3 = exp\{ \tilde \beta_3\} /Z,\\
 \pi_4 &= exp\{ \tilde \beta_1+\tilde\beta_2\}/Z, \,\, \pi_5 = exp\{  \tilde\beta_1 + \tilde\beta_3 + \tilde\beta_4\}/Z,\,\,
\pi_6 = exp\{  \tilde\beta_2 +\tilde \beta_3 + \tilde\beta_5\} /Z,\\
\pi_7 &= exp\{ \tilde \beta_1 + \tilde\beta_2 + \tilde\beta_3 + \tilde\beta_4 + \tilde\beta_5\}/Z 
 \end{align*}  
where $Z = \sum_{i \in \mathcal{I}} p(\boldsymbol i) = exp\{\tilde\beta_1\}   + exp\{\tilde \beta_2\}+\dots +  exp\{ \tilde \beta_1 + \tilde\beta_2 + \tilde \beta_3 + \tilde\beta_4 + \tilde\beta_5\}$ plays the role of the normalization constant. Therefore, for an appropriate parameter $\boldsymbol \gamma = \boldsymbol \gamma(\boldsymbol \pi)$,  
$$\log \boldsymbol \pi = \mathbf{U}_{[AC][BC]} \boldsymbol \gamma,$$
where the design matrix is equal to
\begin{equation}\label{modQHomAS3}
\mathbf{U}_{[AC][BC]} = \left(\begin{array}{ccccccc}
1 & 1 & 1 & 1 & 1 & 1 & 1\\
1 & 0 & 0 & 1 & 1 & 0 & 1\\ 
0 & 1 & 0 & 1 & 0 & 1 & 1 \\ 
0 & 0 & 1 & 0 & 1 & 1 & 1 \\ 
0 & 0 & 0 & 0 & 1 & 0 & 1 \\ 
0 & 0 & 0 & 0 & 0 & 1 & 1 \\
\end{array}\right).
\end{equation}
Notice that removing the first row from $\mathbf{U}_{[AC][BC]} $ leads to the design matrix  $\mathbf{S}_{[AC][BC]}$  given in (\ref{modMxCIAS}), meaning that the quasi model is obtained from the HAS model by including the overall effect. In fact,  the conditional distribution (\ref{qCI}) can also be expressed as:      
$$\boldsymbol \pi = \frac{\exp\{\mathbf{S}_{[AC][BC]}\tr\boldsymbol \beta\}}{ \boldsymbol 1\tr \exp\{\mathbf{S}_{[AC][BC]}\tr\boldsymbol \beta\}}.$$
 
A kernel basis matrix of the quasi $[AC][BC]$ model is equal to 
\begin{equation}\label{kernelCI_OE}
\left(\begin{array}{cccccccc}
0 & 0& 1&  0& -1 &-1 & 1\\
\end{array} \right),
\end{equation}
and thus, the model is specified by a single conditional odds ratio:
$${COR}(AB \mid C = 1)= \frac{p_{001}p_{111}}{p_{011}p_{101}} = 1.$$
The conditional odds ratio ${COR}(AB \mid C = 0)$ is not restricted by the quasi model, so this model characterizes the joint action of $A$ and $B$ when $C$ is present, but provides no restriction otherwise.  The model has one less degree of freedom than the conventional log-linear model $[AC][BC]$ on the CP or the AS conditional independence on the IP.   
\vspace{2mm}

Following the same approach, one can construct a model matrix for the quasi model of homogeneous association, $[AB][AC][BC]$, generated by the ascending class $\mathcal{F}_{asc} = \{\{ABC\}\}$, based on the matrix $\mathbf{S}_{[AB][AC][BC]}$ in (\ref{modMxHomAss}):
\begin{equation}\label{modQHomAS3}
\mathbf{U}_{[AB][AC][BC]} = \left(\begin{array}{ccccccc}
1 & 1 & 1 & 1 & 1 & 1 & 1\\
1 & 0 & 0 & 1 & 1 & 0 & 1\\ 
0 & 1 & 0 & 1 & 0 & 1 & 1 \\ 
0 & 0 & 1 & 0 & 1 & 1 & 1 \\ 
0 & 0 & 0 & 1 & 0 & 0 & 1 \\ 
0 & 0 & 0 & 0 & 1 & 0 & 1 \\ 
0 & 0 & 0 & 0 & 0 & 1 & 1 \\
\end{array}\right).
\end{equation}
This matrix is full-rank, and therefore, the quasi model generated by $\mathcal{F}_{asc} = \{\{ABC\}\}$ is the saturated model.  It does not impose any restriction on the parameter space, and is simply a reparameterization of the exponential family using a different design matrix. \qed

\vspace{5mm}

The next theorem summarizes the relationship between the LL, AS, and quasi LL models:

\begin{theorem}\label{allmodels}
Let $\mathcal{F}_{asc} \subseteq \mathcal{V}$ be an ascending class of subsets of $\mathcal{F}$, and $\mathcal{F}_{des} = \mathscr{F} \mysetminusA \mathcal{F}_{asc} \mysetminusA \{\varnothing\}$. Let $\mathbf{A}$ be the $|\mathcal{F}_{des}| \times I$ indicator matrix of the elements of $\mathcal{F}_{des}$, and denote by $\mathbf{A}_1$ and $\mathbf{A}_0$ the matrices obtained as:
$${\mathbf{A}}_1 = \left(\begin{array}{c}  \boldsymbol 1\tr \\
 \mathbf{A} \end{array} \right), \qquad {\mathbf{A}}_0 = \left(\begin{array}{cc} 1 & \boldsymbol 1\tr \\
\boldsymbol 0 & \mathbf{A} \end{array} \right),$$
with respective sizes $(|\mathcal{F}_{des}|+1) \times I$ and $(|\mathcal{F}_{des}|+1) \times (I+1)$, and $\boldsymbol 0$ and $\boldsymbol 1$ being the column vectors of appropriate dimensions.
Then, the following  holds:
\begin{enumerate}[(i)]
\item $HAS(\mathcal{F}_{asc})$ and $QLL(\mathcal{F}_{asc})$ are the relational models on the IP $\mathcal{I}$, generated by the matrices $\mathbf{A}$ and $\mathbf{A}_1$, respectively.
\item $LL(\mathcal{F}_{asc})$ is the relational model on the CP $\mathcal{I}_0$, generated by the matrix  $\mathbf{A}_0$. 
\item Both $HAS(\mathcal{F}_{asc})$ and $LL(\mathcal{F}_{asc})$ models have $|\mathcal{F}_{asc}|$ degrees of freedom. The  number of the degrees of freedom of  quasi model $QLL(\mathcal{F}_{asc})$ is equal to $|\mathcal{F}_{asc}|-1$. 
\end{enumerate}
\end{theorem}

\noindent The statements follow from the results obtained in \cite{KRD11, KRoverEff}. The theorem also implies that 

\begin{corollary}
The odds ratio representation of $HAS(\mathcal{F}_{asc})$ can be obtained from the one for $LL(\mathcal{F}_{asc})$, by taking out the probability of the zero cell $p(\boldsymbol 0)$. 
The odds ratio representation of $QLL(\mathcal{F}_{asc})$ can be obtained from one for $LL(\mathcal{F}_{asc})$, by eliminating constraints on the odds ratios containing $p(\boldsymbol 0)$.
\end{corollary}

A diagram summarizing Theorem \ref{allmodels} is shown in Figure \ref{DiagramStat}.

\begin{example} \label{drugs}
The study described in \cite*{DrugInjection2013} aimed to estimate the number of current injecting drug users in Scotland in 2006, based on four data sources: social enquiry reports (DS1), hospital records (DS2), drug treatment agencies (DS3), recent HCV diagnoses (DS4).  The records of 5670 individuals were obtained from various combinations of those sources and are summarized in Table 1 in \cite{DrugInjection2013}.

It is clear from the context, that although the subjects who do not inject any drugs of interest do exist, these subjects are not going to be listed in any of the data sources described above. Thus, the sample space is an IP.

The model search within the HAS model class identified the HAS model of no third order interaction, [DS1 DS2 DS3 DS4],  as the best fitting model, with the Pearson chi-square statistic equal to $ 5.81$ on one degree of freedom. 

The quasi models, specified by the setting one of the second order interactions to zero, fit well, with the Pearson statistic around $1$ on one degree of freedom. Note, however, that the fit of a quasi-model means that it accounts for about half of the data and the other half remains unmodeled. \qed
\end{example}

In the next section, the relationship between the three model classes is given a geometrical interpretation.

\section{Geometric view}\label{SectionGeometry}

Usually, a statistical model  is seen as the collection of probability distributions possessing particular properties.  These properties can often be formulated using polynomial equations on the log-scale that are satisfied by all distribution in the model. In many cases,  the polynomial representation is not unique, and thus, there may be other properties of the distributions which lead to the same model. Clearly, a model cannot be fully understood without taking into account all properties which could potentially define it.  These considerations suggest using the approach of algebraic geometry for statistical modeling. Under this approach, the concept of polynomial variety would be used to define the collection of probability distributions in the model, using one set of restrictions to define the model, and the notion of vanishing ideal may be interpreted as the collection of all possible interpretations, that is, all possible specifications, of this model.  A more formal introduction is given next.  As all three model classes of interest are special cases of the relational models, the relational model framework is used.

Let $\mathcal{T}$ be a sample space, for example,  $\mathcal{T}=\mathcal{I}_0$ (CP) or $\mathcal{T}=\mathcal{I}$ (IP). Write $T=|\mathcal{T}|$ for the number of cells in $\mathcal{T}$, and denote by  $int(\Delta_{T-1})$ the interior of the $(T-1)$-dimensional simplex. Let $\mathbf{A}$ be a 0-1 matrix of full row rank and no zero columns. Then, the relational model with a design matrix $\mathbf{A}$ is the set of probability distributions comprising
\begin{equation}\label{ExtMprobnb}
int(\Delta_{T-1}) \cap \mathcal{X}_{\mathbf{A}},
\end{equation}
where $\mathcal{X}_{\mathbf{A}}$ is the polynomial variety associated with  $\mathbf{A}$: 
\begin{equation}\label{variety}
\mathcal{X}_{\mathbf A} = \left\{\boldsymbol p \in \mathbb{R}^{T}_{\geq 0}: \,\, \boldsymbol p^{\boldsymbol d^+} = \boldsymbol p^{\boldsymbol d^-}, \,\,  \forall  \boldsymbol d \in Ker(\mathbf{A}) \right \}.
\end{equation}
Here $\boldsymbol{d}^{\boldsymbol{+}}$ and $\boldsymbol{d}^{\boldsymbol{-}}$ denote, respectively, the positive and negative parts of $\boldsymbol d$,  and $ \boldsymbol{p}^{\boldsymbol{d}^{\boldsymbol{+}}}$ and $ \boldsymbol{p}^{\boldsymbol{d}^{\boldsymbol{-}}}$ are interpreted componentwise.  In fact, in order to be in the model,  for $\boldsymbol p \in int(\Delta_{T-1})$ it is sufficient to require that the condition $\boldsymbol{p}^{\boldsymbol{d}^{\boldsymbol{+}}} = \boldsymbol{p}^{\boldsymbol{d}^{\boldsymbol{-}}}$  holds for a set of $\boldsymbol d$ comprising a basis of $Ker(\mathbf{A})$, and it can also be replaced with the equivalent requirement on the generalized odds ratios: $\boldsymbol{p}^{\boldsymbol{d}^{\boldsymbol{+}}} / \boldsymbol{p}^{\boldsymbol{d}^{\boldsymbol{-}}} = 1$.

Let $\bd 1\tr$ = $(1, \dots,1)$ be the row of $1$'s of length $T$.  If $\boldsymbol 1\tr$ belongs to the space spanned by the rows of $\mathbf A$,  the relational model is said to have the overall effect. In this case, the toric variety $\mathcal{X}_{\mathbf A}$ is generated by homogeneous polynomials, and, correspondingly, the relational model can be specified using only homogeneous generalized odds ratios.  Relational models without the overall effect are specified using homogeneous and at least one non-homogeneous generalized odds ratios \citep{KRD11}, and the corresponding variety is affine. The presence of the overall effect is an intrinsic property of the model which does not pertain to a particular parameterization, and relational models with and without the overall effect have very different characteristics \citep{KRD11}.  Some consequences of adding and removing the overall effect to a relational model were studied in detail in \cite{KRoverEff}. Because LL and QLL models are models with an overall effect, their polynomial varieties are homogeneous.  

Let $\mathcal{F}_{asc}$ be an ascending class of subsets of features $\mathcal{F}$ on the IP $\mathcal{I}$. As shown in Section \ref{SectionNoOE}, the HAS model is a relational model with the design matrix $\mathbf{A}$ whose rows are indicators of the subsets in the descending class $\mathcal{F}_{des}$, and, by Theorem \ref{theoremNoOE}, this model does not have an overall effect, and thus, the variety $\mathcal{X}_{\mathbf A}$ is non-homogeneous. The polynomial ideal, $\mathscr{I}_{\mathbf{A}}$, associated with $\mathbf{A}$ is generated by the polynomials defining the variety $\mathcal{X}_{\mathbf{A}}$ (the latter is the zero set of the former) is also non-homogeneous. A procedure of homogenization of an ideal by introducing a new variable is described, for example, in \cite{Cox}. In each non-homogeneous polynomial generating the ideal, the monomials are multiplied by a (non-negative) power of the new variable, say $p_0$, so that their degrees become equal to the degree of the original polynomial \citep[cf.][p.400, Proposition 8.2.7]{Cox}. The homogeneous polynomials and the corresponding homogeneous odds ratios stay unchanged. For details of this procedure in the case of relational models see \cite{KRoverEff}.  The polynomial ideal generated by the homogenized polynomials and the corresponding polynomial variety are homogeneous. The minimal variety that can be obtained after homogenization is called the projective closure of the affine variety $\mathcal{X}_{\mathbf A}$ \citep[cf.][p.418, Definition 8.4.6]{Cox}. The latter can be obtained from the former by dehomogenization via setting $p_0 = 1$. The homogenization of AS independence was considered in \cite{KRoverEff}. Theorem 4 therein established that the intersection of the projective closure of the model of AS independence  on the sample space $\mathcal{I}$ with $int(\Delta_{I-1})$ is the traditional log-linear model  of independence on the sample space $\mathcal{I}_0$. The idea used in the proof is applied here to prove a more general result, stated next.


\begin{theorem}
Let $\mathcal{F}_{asc}$ be an ascending class in $\mathscr{V}$. The intersection of $int(\Delta_{I})$ and the projective closure of the variety corresponding to the HAS model generated by $\mathcal{F}_{asc}$ on the IP $\mathcal{I}$ is homeomorphic to the $LL(\mathcal{F}_{asc})$ model in the CP $\mathcal{I}_0$.
\end{theorem}

\begin{proof}
 Let $\mathcal{F}_{des} = \mathscr{V} \mysetminusA \mathcal{F}_{asc} \mysetminusA \{\varnothing\}$ and $\mathcal{F}_{des,0} = \mathcal{F}_{des}\cup \{\varnothing\}$. Let $\mathbf{A}$ and $\mathbf{A}_0$ be the indicator matrices for the cylinder subsets of the marginal distributions in the class $\mathcal{F}_{des}$ (on the IP) and $\mathcal{F}_{des,0}$  (on the CP), respectively. Notice that 
$${\mathbf{A}}_0 = \left(\begin{array}{cc} 1 & \boldsymbol 1\tr \\
\boldsymbol 0 & \mathbf{A} \end{array} \right).$$
The intersection of the polynomial variety $\mathcal{X}_{\mathbf{A}_0}$ with the interior of the $I$-dimensional simplex is the set of distributions on $\mathcal{I}_0$ comprising the log-linear model  generated by $\mathcal{F}_{des,0}$.   

Because the row $\boldsymbol 1\tr$ is not a linear combination of rows of $\mathbf{A}$, the polynomial variety $\mathcal{X}_{\mathbf{A}}$ is non-homogeneous \citep[cf.][]{SturBook}.  In order to obtain the projective closure of $\mathcal{X}_{\mathbf{A}}$ \citep[cf.][p.419, Theorem 8.4.8]{Cox}, include the zero cell, $\boldsymbol i_0$, to the sample space, choose a Gr{\"o}bner basis of the ideal $\mathscr{I}_{\mathbf{A}}$, and homogenize all non-homogeneous polynomials in this basis using the cell probability $p({\boldsymbol i_0})$.  By Proposition 8.4.7 in \cite[][p.418]{Cox}, the projective closure of $\mathcal{X}_{\mathbf{A}}$ is the minimal homogeneous variety in the projective space whose dehomogenization is $\mathcal{X}_{\mathbf{A}}$. By Theorem 3(ii) in \cite{KRoverEff}, this projective closure is the variety $\mathcal{X}_{\mathbf{A}_0}$.

Each distribution in the relational model with the model matrix $\mathbf{A}$ has the multiplicative structure prescribed by $\mathbf{A}$ \citep{KRextended}, and during the homogenization, is mapped in a positive distribution in $\mathcal{X}_{\mathbf{A}_0}$. All strictly positive distributions in $\mathcal{X}_{\mathbf{A}_0}$ have the multiplicative structure prescribed by $\mathbf{A}_0$, and thus, the intersection of $\mathcal{X}_{\mathbf{A}_0}$ with the interior of the $I$-dimensional simplex is the log-linear model  generated by $\mathbf{A}_0$ (equivalently, by the descending class $\mathcal{F}_{des,0}$).   
\end{proof}

The next example illustrates homogenization of the HAS conditional independence.

\vspace{5mm}

\noindent \textbf{Example \ref{k=3t}} (revisited):

The polynomial ideal corresponding to the HAS model of conditional independence $[AC][BC]$ on the IP $\mathcal{I}$ is generated by the binomials $p_{110}-p_{010}p_{100}$ and $p_{001}p_{111}-p_{011}p_{101}$. The first binomial is non-homogeneous and can be homogenized by including a new variable, say $p_0$, in the first term. The homogenized pair:
$$p_0\,p_{110}-p_{010}p_{100}, \qquad \mbox{and} \qquad  p_{001}p_{111}-p_{011}p_{101},
$$
generates a homogeneous ideal in $int(\Delta_7)$ whose variety is the projective closure of the original variety of the HAS model.  Note that after homogenization, the first equation seizes to restrict the cell probabilities on the IP. Further, if the first polynomial is removed, the specification of the quasi model is obtained.
\qed

\vspace{8mm}

A schematic description of the results of this section is given in Figure \ref{DiagramGeo}.



Both the  LL and the QLL models are examples of so-called regular multiliear monomial models \cite[cf.][]{LeonelliRico}. A model closely related to the HAS independence on two features was considered in \cite{LeonelliRico}, Example 1. In their terminology, HAS models belong to the class of monomial discrete parametric models. Although HAS models can be specified by square-free monomials, they are neither regular nor multilinear.

\section{Maximum likelihood estimation and model selection}\label{SectionModelSelection}

The absence of the overall effect in HAS models on IP has a strong influence on the properties of the maximum likelihood estimates. If the overall effect is not present, then under multinomial sampling the observed subsets sums are only preserved by the MLE up to a constant multiplier and under Poisson sampling the observed total is not preserved, which is in sharp contrast with the behaviour of log-linear models on the CP \citep{KRD11}. 

The detailed results about the maximum likelihood estimation under relational models, obtained in \cite{KRipf1} and \cite{KRextended}, also hold for the hierarchical AS models. 

Relational models for probabilities without the overall effect are curved  exponential families, and the maximum likelihood estimates under such models have specific characteristics. In fact, if $\boldsymbol{q}$ is the observed probability distribution, and $\boldsymbol{r} = \boldsymbol 1$ is the vector of $1$'s, obviously satisfying $\mathbf{D} \log \boldsymbol 1 = \boldsymbol 0$, then the distribution $\boldsymbol{p}$ given in Theorem \ref{TwoDistr1} is the MLE. This implies that if unique multiplicative parameters are allowed in certain cells, the MLEs in these cells do not necessarily reproduce the observed probabilities in these cells, thus model misfit will not be concentrated to the other cells and the standard approach to defining QLLs does not work. Relational models without the overall effect are not scale invariant, and the MLEs for the cell frequencies do not necessarily satisfy the multiplicative structure prescribed by the model for the cell probabilities. Another implication is that the condition $\boldsymbol 1'\boldsymbol p = 1$ cannot be achieved through a normalization constant (the overall effect), but has to be treated as a separate model constraint.   The MLEs can be computed, for instance, using the algorithms of \cite{EvansForcina11} and \cite{Forcina2019}. A generalized iterative proportional fitting procedure, G-IPF, that can be used for both models with and models without the overall effect was proposed in \cite{KRipf1} and is implemented in \cite{gIPFpackage}. 

The asymptotic properties of the MLE under relational models were studied in \cite{KRtesting}. It was shown that the maximum likelihood estimates for the cell probabilities are asymptotically normally distributed. It was also proved that both the Pearson and the likelihood ratio statistics have  chi-squared distribution with the number degrees of freedom equal to the number of the degrees of freedom of the model.

A model selection procedure of \cite{EdwardsHavranek}, proposed for the graphical and hierarchical log-linear models, uses the structure of generating cut and does not rely on the presence of the overall effect. Therefore, this procedure can be applied for the model search in the class $HAS$.

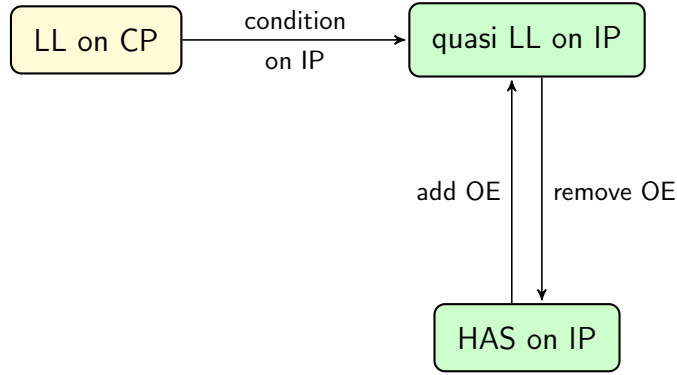
\begin{figure}
\centering
\begin{tikzpicture}[
  font=\sffamily,
  every matrix/.style={ampersand replacement=\&,column sep=3cm,row sep=3cm},
  source/.style={draw,thick,rounded corners,fill=yellow!20,inner sep=.3cm},
  process/.style={draw,thick,circle,fill=blue!20},
  sink/.style={source,fill=green!20},
  datastore/.style={draw,very thick,shape=datastore,inner sep=.3cm},
  dots/.style={gray,scale=2},
  to/.style={->,>=stealth',shorten >=1pt,semithick,font=\sffamily\footnotesize},
  every node/.style={align=center}]

  \matrix{
    \node[source] (hisparcbox) {LL on CP};
    \& \node[sink] (daq) {quasi LL on IP}; \\
    \& \node[sink] (buffer) {HAS on IP}; \\
  };

  \draw[to] (hisparcbox) -- node[midway,above] {condition}
      node[midway,below] { on IP} (daq);
  \draw[to] ([xshift=2mm]daq.south) -- node[midway,right] {remove OE} ([xshift=2mm]buffer.north);
  \draw[to] ([xshift=-2mm]buffer.north) -- node[midway,left] {add OE} ([xshift=-2mm]daq.south);
\end{tikzpicture}
\caption{Interpretation in statistical terms.}
\label{DiagramStat}
\end{figure}

\vspace{10mm}

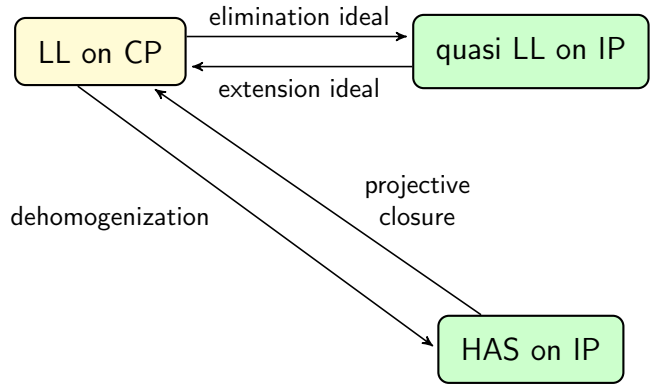
\begin{figure}
\centering
\begin{tikzpicture}[
  font=\sffamily,
  every matrix/.style={ampersand replacement=\&,column sep=3cm,row sep=3cm},
  source/.style={draw,thick,rounded corners,fill=yellow!20,inner sep=.3cm},
  process/.style={draw,thick,circle,fill=blue!20},
  sink/.style={source,fill=green!20},
  datastore/.style={draw,very thick,shape=datastore,inner sep=.3cm},
  dots/.style={gray,scale=2},
  to/.style={->,>=stealth',shorten >=2pt,semithick,font=\sffamily\footnotesize},
  every node/.style={align=center}]

  \matrix{
    \node[source] (hisparcbox) {LL on CP};
    \& \node[sink] (daq) {quasi LL on IP}; \\
    \& \node[sink] (buffer) {HAS on IP}; \\
  };

   \draw[to] ([yshift=2mm]hisparcbox.east) -- node[midway,above] {elimination ideal}([yshift=2mm]daq.west);
   \draw[to] ([yshift=-2mm]daq.west) -- node[midway,below] {extension ideal} ([yshift=-2mm]hisparcbox.east);
   \draw[to] ( buffer) -- node[midway,right=0.5cm] { projective\\closure} ( hisparcbox);
   \draw[to] ([xshift=-3mm]hisparcbox.south) -- node[midway,left=0.5cm] { dehomogenization} (buffer.west);
\end{tikzpicture}

 \caption{Interpretation in terms of Algebraic Geometry.}
 \label{DiagramGeo}
  \end{figure}

\section*{Acknowledgements}

The authors thank Antonio Forcina for his insightful comments.

\bibliographystyle{apacite}
\bibliography{cells}

\end{document}